\documentclass{llncs}

\usepackage{amsfonts,amsmath,amssymb}
\usepackage{mathtools}
\usepackage{algorithm}
\usepackage{graphicx}
\usepackage{color}
\usepackage[normalem]{ulem}
\usepackage{multirow}
\usepackage{pstricks}
\usepackage{subfig}

\author{Peter~Bulychev\inst{1}, Franck~Cassez\inst{2}, Alexandre~David\inst{1},
  Kim~G.~Larsen\inst{1}, Jean-Fran\c{c}ois~Raskin\inst{3}, Pierre-Alain Reynier\inst{4}}

\institute{	
   CISS, CS, Aalborg University, Denmark\\
   \texttt{\{pbulychev,adavid,kgl\}@cs.aau.dk}\\ 
   \and
   National ICT Australia, Sydney, Australia \\
   \texttt{Franck.Cassez@nicta.com.au} \\ 
   \and
   Computer Science Department, Universit\'e Libre de Bruxelles (U.L.B.), Belgium\\
   \texttt{jraskin@ulb.ac.be}\\
   \and 
   LIF, Aix-Marseille University \& CNRS, France\\
   \texttt{pierre-alain.reynier@lif.univ-mrs.fr}\\
   }

\title{
Controllers with Minimal Observation Power \\
(Application to Timed Systems)\thanks{Partly supported by  ANR project 
ECSPER (JC09- 472677), ERC Starting Grant inVEST-279499, Danish-Chinese Center for Cyber Physical Systems (IDEA4CPS) and VKR Center of Excellence MT-LAB.}
}

\newcommand{\setR}{\ensuremath{\mathbb{R}}}

\begin{document}

\maketitle

\vspace{-10pt}

\begin{abstract}
We consider the problem of controller synthesis 
under imperfect information in a setting where 
there is a set of available observable predicates equipped with a cost function.
The problem that we address 
is the computation of a subset of predicates sufficient for control 
and whose cost is minimal. Our solution avoids a full exploration
of all possible subsets of predicates 
%, uses an on-the-fly algorithm for solving the underlying games with imperfect information 
and reuses some information between different iterations. 
We apply our approach to timed systems. We have developed a tool prototype and 
analyze the performance of our optimization algorithm on two case studies. 
\end{abstract}

\section{Introduction}

Timed automata by Alur and Dill~\cite{Alur94atheory} is one of the
most popular formalism for the modeling of real-time systems.  
%The literature on timed automata contains a number of successful applications that have motivated progresses both in theory and tool support.  
One of the applications of Timed Automata is {\em controller
  synthesis}, i.e. the automatic synthesis of a controller strategy
that forces a system to satisfy a given specification.  
%There is a large body of works that uses the game metaphor and games played on graphs for solving the controller synthesis problem, see e.g.~\cite{Thomas95}. 
For timed systems, the controller synthesis
problem has been first solved in~\cite{MalerPS95} and
progress on the algorithm obtained in~\cite{cdfll05} has made possible
the application on examples of a practical interest.
This algorithm has been implemented in the {\sc Uppaal-Tiga}
tool~\cite{bcdfll07}, and applied to several case
studies~\cite{guided_controller_synthesis,automatic_synthesis,flexible_timeline_based,tga_based_controllers}.

The algorithm of~\cite{cdfll05} assumes that the controller has {\em perfect information} about the evolution of the system during its execution.  
%This assumption is an obstacle for many applications. 
However, in practice, it is common that the controller
acquires information about the state of the system via a finite set of sensors each of 
them having only a finite precision. 
This motivates to study \emph{imperfect information} games.
% or, in other contexts, the controller has only access to the public variables of the process that it controls and not to the values of its private variables for example. 
%In this case we say that the controller has {\em imperfect information} only.
%\footnote{In the literature, the terms {\em partial information} and {\em imperfect information} are also used.}.
%The reason for that is that in the
%real life the controllers build their
%decisions based on the sensor's values, thus the controllers have only imperfect (or partial) information
%on the state of the environment.  

%A framework for controller synthesis problems under partial observation is provided by
%games with {\em imperfect information}. 
%In a game graph with incomplete information, the states of the game
%graph are partitioned into equivalence classes, and in the course of
%the game, {\em only} the equivalence class to which belongs the
%current state of the game is revealed. So the controller must take its
%decisions based on that information only: the controller decides what
%to play according to the sequence of equivalence classes observed
%since the beginning of the game. 
The first theoretical results on \emph{imperfect information} games have been obtained
in~\cite{reif84}, followed by algorithmic progresses and
  additional theoretical results in~\cite{RaskinCDH07}, as well as
  application to timed games in~\cite{BouyerDMP03,cdllr07}. This
paper extends the framework of~\cite{cdllr07} and so we consider the
notion of {\em stuttering-invariant observation-based strategies}
where the controller makes choice of actions only when changes in its observation occur. 
The observations are defined by the values of a finite set of \emph{observable state predicates}. 
Observable predicates correspond, for example, to
information that can be obtained through sensors by the controller.
%, and it's assumed that the controller observes only
%the changes in observations, thus its strategy should be
%stuttering-invariant.  
In~\cite{cdllr07}, a symbolic algorithm for computing
  observation-based strategies for a \emph{fixed} set of observable predicates
%the synthesis of such  strategies 
is proposed, and this algorithm has been implemented in
  {\sc Uppaal-Tiga}.

  %In~\cite{cdllr07}, it is assumed that the set of observable  predicates is fixed, and the algorithm solves the game of imperfect  information for this set of observable predicates only.  
  
  In the   current paper, we further develop the approach of~\cite{cdllr07} and consider %the
  %case when there is 
  a set of \emph{available} observation predicates equipped with a
  cost function. Our objective is to synthesize a winning strategy that uses a
  subset of the available observable predicates with a {\em minimal
    cost}.  Clearly, this can be useful in the design process when we
  need to select sensors to build a controller.
%This can be used to model the following design problem: several sensors with different characteristics on the
%market exists, and we want to synthesize a controller that uses the cheapest set of sensors.
%%with the minimal cost.

Our algorithm works by iteratively picking different
subsets of the set of the available observable predicates, solving the game
for these sets of predicates and finally finding the controllable
combination with the minimal cost.  Our
algorithm avoids the exploration of all possible combinations by
taking into account the inclusion-set relations between different sets of observable
predicates and monotonic properties of the underlying games. 
Additionally, for efficiency reasons, our algorithm reuses, when solving the game for a new set of observation predicates,  
information computed on previous sets whenever possible. 
% (i.e. when it solves the game for the different sets of observable predicates).
%For instance, if we figured out that the controller loses
%for some set of observable predicates, then we can conclude that it
%will also loose for all coarser observations (since it will have less
%information and thus less controllability).  Additionally, we can
%interpret a game with a coarser set of observable predicates as an
%imperfect information game with respect to a game with finer
%observations.  Thus our algorithm reuses the state spaces of the finer
%games when it solves more coarse games.

%This algorithm has been implemented in a prototype tool, and we
%analyze the effectiveness of various parameters of the algorithm by
%applying this tool to two case studies.
  
\vspace{-5pt}

\paragraph{{\bf Related works}}
Several works in the literature consider the synthesis of controllers
along with some notion of
optimality~\cite{DBLPconffsttcsBouyerCFL04,BrihayeBR05,journals/fmsd/BouyerBBR07,ChatterjeeHJS11,DBLPconfhybrid2009,Zwick96thecomplexity,ChatterjeeMH08,MNR-ATVA11}
but they consider the minimization of a cost along the execution of
the system while our aim is to minimize a static property of the
controller: the cost of observable predicates on which its winning
strategy is built.  The closest to our work is~\cite{ChatterjeeMH08}
where the authors consider the related but different problem of
turning on and off sensors during the execution in order to minimize
energy consumption.  In~\cite{HenzingerJM03}, the authors consider
games with perfect information but the discovery of interesting
predicates to establish controllability. In~\cite{DimitrovaF08} this
idea is extended to games with imperfect information. In those two
works the set of predicates is not fixed a priori, there is no cost
involved and the problems that they consider are undecidable.
In~\cite{MNR-ATVA11}, a related technique is used: a hierarchy on different levels of abstraction is considered, which allows to use analysis done on coarser abstractions to reduce the state space to be explored for more precise abstractions.

%literature solve optimization problems on timed automata~\cite{journals/fmsd/BouyerBBR07,DBLPconffsttcsBouyerCFL04} but their consider the minimization  
%In~\cite{journals/fmsd/BouyerBBR07}, the authors 
%It should be noted, that our paper is in some sense orthogonal to the related works on automata-based optimal controller synthesis~\cite{journals/fmsd/BouyerBBR07,DBLPconffsttcsBouyerCFL04,ChatterjeeHJS11,DBLPconfhybrid2009,Zwick96thecomplexity}.
%The difference is that we synthesize the minimal observation power, i.e. we are optimizing the static cost of the controller.
%In contrast, the previous papers optimized the worst (or the average) cost of a system execution, e.g. in terms of energy being consumed by sensors and actuators (that both can be turned and off). 
%Thus the previous papers optimized the dynamic (run-time) characteristics of the controller instead of the static one.

%\marginpar{May be this paragraph can be suppressed because we do not
%  comment precisely on the differences in notations. FC: OK for me} It
%should be noted, that although our paper is based on~\cite{cdllr07},
%we use here slightly different notations, for two reasons.  First, it
%simplifies the presentation of the lattice-based algorithm.  And
%second, it allows us to present the algorithm in a more abstract form
%that can be applied to any two-player games of incomplete information
%and not only timed games.  Additionally in this paper, we limit
%ourselves to safety objectives, while~\cite{cdllr07} studies
%reachability games as well.

\vspace{-5pt}

\paragraph{{\bf Structure of the paper}} In section~\ref{lts}, we
define a notion of {\em labeled transition systems} that serves as the
underlying formalism for defining the semantics of the two-player
safety games. In the same section we define imperfect information
games and show the reduction of~\cite{reif84} of these games to
the games with complete information. % (this reduction has been firstly described
%in~\cite{reif84} for discrete games and in \cite{cdllr07} for Timed
%Game Automata).  
Then in section~\ref{tga} we define {\em timed game automata}, that we use as a modeling
formalism. In section~\ref{problemstatement}, we state
the cost-optimal controller synthesis problem and show that a natural
extension of this problem (that considers a simple infinite set of
observation predicates) is undecidable.  In
section~\ref{lattice-based}, we propose an algorithm and in section~\ref{casestudies}, we present two case studies.

\vspace{-10pt}
%\section{Partial Information Games on Labeled Transition
%  Systems}\label{lts}
\section{Games with Incomplete Information}
%on Labeled Transition
%  Systems}
\label{lts}

%\vspace{-10pt}

%The first player in these games suggests an action to play, and the second player resolves nondeterminism in LTS by choosing a transition labeled with this action.

%Then we describe a reduction from the stuttering-invariant strategy synthesis problem to the solution of a fully observable two-player game, that is also called a \emph{knowledge} game.

%The problem of the cost-optimal controller synthesis is decidable when the knowledge game is finite state.
%The results of this paper can be applied to the case when the knowledge game is finite state
%This is obviously true for the case when the original partial information game is finite-state.
%Additionally, as we have demonstrated in~\cite{cdllr07}, the knowledge game is finite state also for the case, when the original 

%Finally, we define the syntax and the LTS-based semantics of Timed Game Automata, that we use as a modeling formalism in our paper.

\subsection{Labeled Transition Systems}

\begin{definition}[Labeled Transition System]\label{def-lts}
  A \emph{Labeled Transition System} (LTS) $A$ is a tuple $(S,
  s_{init} , \Sigma, \rightarrow)$ where:
\begin{itemize}
\item $S$ is a (possibly infinite) set of states,
\item $s_{init} \in S$ is the initial state,
\item $\Sigma$ is the set of actions,
\item $\rightarrow \subseteq S \times \Sigma \times S$ is a transition
  relation, we write $s_1 \xrightarrow{a} s_2$ if $(s_1,
  a, s_2) \in \rightarrow$. 
\end{itemize}
W.l.o.g. we assume that a transition relation is total, i.e. for all
states $s \in S$ and actions $a \in \Sigma$, there exists $s' \in S$
such that $s \xrightarrow{a} s'$.
\end{definition}

%A state predicate $\varphi$ is a boolean-valued function defined on
%the set of states $S$.  We write $s \models \varphi$ when $\varphi$ is
%satisfied by a state $s$, i.e. $\varphi(s) = true$.  

A \emph{run} of a LTS is a finite or infinite sequence of states
$r=(s_0, s_1, \dots, s_n, \dots)$ such that 
%pairs of consecutive states, we have 
$s_i \xrightarrow{a_i} s_{i+1}$ for some action $a_i \in \Sigma$.
$r^i$ denotes the prefix run of $r$ ending at $s_i$. We denote
by $Runs(A)$ the set of all finite runs of the LTS $A$ and by
$Runs^{\omega}(A)$ the set of all infinite runs of the LTS $A$.

A \emph{state predicate} is a characteristic function $\varphi: S
\rightarrow \{0,1\}$.  We write $s \models \varphi$ iff
$\varphi(s)=1$.

We use LTS as arenas for games: at each round of the game Player I
(Controller) chooses an action $a \in \Sigma$, and Player II (Environment) resolves the
nondeterminism by choosing a transition labeled with $a$. Starting
from the state $s_{init}$, the two players play for an infinite number
of rounds, and this interaction produces an infinite run that we call
the {\em outcome} of the game. The objective of Player I is to keep
the game in states that satisfy a state predicate $\varphi$, this
predicate typically models the safe states of the system.

More formally, 
%when engaged in a game against Player II on the arena $A$
Player I plays according to a strategy $\lambda$ (of Player I)
which is a mapping from the set of finite runs to the set of actions,
i.e. $\lambda : Runs(A) \rightarrow \Sigma$. % and such that for all finite runs $(s_0, s_1,\dots,s_n)$, if $\lambda(s_0, s_1,\dots,s_i)=a$ then there exists a transition $s_{i} \xrightarrow{a} s_{i+1}$ in $A$. 
We say that an infinite run $r = (s_0,s_1, s_2, \dots, s_n, \dots) \in
Runs^{\omega}(A)$ is {\em consistent} with the strategy $\lambda$, if
for all $0 \leq i$, there exists a transition $s_i
\xrightarrow{\lambda(r^i)} s_{i+1}$. We denote by ${\sf
  Outcome}(A,\lambda)$ all the infinite runs in $A$ that are 
consistent with $\lambda$ and start in $s_{init}$.  An infinite run
$(s_0, s_1,\dots,s_n,\dots)$ {\em satisfies} a state predicate
$\varphi$ if for all $i \geq 0$, $ s_i \models \varphi$.  A (perfect information)
\emph{safety game} between Player I and Player II is defined by a pair
$(A, \varphi)$, where $A$ is an LTS and $\varphi$ is a state predicate
that we call a \emph{safety state predicate}.  The {\em safety game
  problem} asks to determine, given a game $(A, \varphi)$, if there
exists a strategy $\lambda$ for Player I such that all the infinite
runs in ${\sf Outcome}(A,\lambda)$ satisfy $\varphi$.

\subsection{Observation-Based Stuttering-Invariant Strategies}

%Let $r=(s_0, s_1, \dots, s_n)$ be a finite sequence.  We denote
%$r^i=(s_0, s_1, \dots, s_i)$ to be the prefix of $r$ of length $i$.
%Let $S$ be some set.

In the imperfect information setting, Player I observes the state of
the game using a set of {\em observable predicates} $obs =
\{\varphi_1, \varphi_2, \dots, \varphi_m \}$.  An \emph{observation}
is a valuation for the predicates in $obs$, i.e. in a state $s$,
Player I is given the subset of observable predicates that are
satisfied in that
state. %, or in other words it is a subset of $obs$ that consists of all the
This is defined by the function $\gamma_{obs}$:
$$\gamma_{obs} (s) \equiv \{\varphi \in obs\ \mid \ s \models \varphi \}$$

We extend the function $\gamma_{obs}$ to sets of states that satisfy the same set of observation predicates.
So, if all the elements of some set of states $v \subseteq S$ satisfy the same set of observable predicates $o$ (i.e. $\forall s \in v \cdot \gamma_{obs}(s) = o$), then we let $\gamma_{obs}(v)=o$.

In a game with imperfect information, Player I has to play according to {\em observation
  based stuttering invariant strategies} (OBSI strategies for
short). Initially, and whenever the current observation of the system
state changes, Player I proposes some action $a \in \Sigma$ and this
intuitively means that he wants to play the action $a$ whenever this
action is enabled in the system. Player I is not allowed to change
his choice as long as the current observation remains the same. %, so in particular, Player I does not observe the number of the transitions taken since the last observation change. 

An \emph{Imperfect Information Safety Game} (IISG) is defined by a triple $(A, \varphi, obs)$.
%This intuition is formalized by the notion of OBSI strategies below.
%OBSI strategies are formalized below.
%We assume that in the imperfect information setting a Controller has to play according to a stuttering-invariant observation-based strategy\cite{cdllr07}, i.e. such strategy that it:
%\begin{itemize}
% \item makes its decisions based on the histories of observations, 
% \item changes its decision only when the observations are changed. 
%\end{itemize}

%A stuttering-invariant observation based strategy~ makes its decision based on the history of values of finitely many state-predicates.

% predicates that are evaluated to $1$.
%Let  be a set of
%state-based predicates. 
%These predicates can represent the information available observations of the Player I, and in this case ` call them  \emph{observable predicates}.
%, i.e. $\varphi_i : S \mapsto \{0, 1\}$ for any $\varphi_i$.

%We call the value of $\gamma_{obs}(s)$ an observation of a state $s$, and each observation is an element of ${\mathcal P}(obs)$.

Consider a run $r = (s_0, s_1, \dots, s_n)$, and its prefix $r'$
that contains all the elements but the last one (i.e. $r = r' \cdot
s_n$).  A \emph{stuttering-free} projection $r \downarrow obs$ of a
run $r$ over a set of predicates $obs$ is a sequence, defined by the
following inductive rules:
\begin{itemize}
\item if $r$ is a singleton (i.e. $n = 0$), then $r \downarrow obs =
  \gamma_{obs}(s_0)$
 \item else if $n>0$ and $\gamma_{obs}(s_{n-1}) = \gamma_{obs}(s_n)$,
   then $r \downarrow obs = r' \downarrow obs$
 \item else if $n>0$ and $\gamma_{obs}(s_{n-1}) \neq
   \gamma_{obs}(s_n)$, then $r \downarrow obs = r' \downarrow obs
   \cdot \gamma_{obs}(s_n)$
\end{itemize}

\begin{definition}\cite{cdllr07}
A strategy $\lambda$ is called $obs$-Observation Based
Stuttering Invariant ($obs$-OBSI) if for any two runs $r'$ and $r''$
such that $r' \downarrow obs = r'' \downarrow obs$, 
the values of $\lambda$ on $r'$ and $r''$ coincide, i.e. $\lambda(r') = \lambda(r'')$.
\end{definition}

We say that Player I wins in IISG $(A, \varphi,obs)$, if there exists a $obs$-OBSI
strategy $\lambda$ for Player I such that all the infinite runs in
${\sf Outcome}(A,\lambda)$ satisfy $\varphi$.

\subsection{Knowledge Games}

The solution of a IISG $(A, \varphi, obs)$
can be reduced to the solution of a \emph{perfect information} safety
game $(G, \psi)$, whose states are sets of states in $A$ and represent the \emph{knowledge} (beliefs) of Player I about the current possible states of $A$. 

%We call the
%states of $G$ \emph{the knowledges} of Player I in $A$, and $G$
%satisfies the following important property.
%If Player I plays the same OBSI strategy and sees the same
%observations in $A$ by $G$, then the current state (knowledge) in $G$
%consists of all the states of $A$ in which it can be after playing
%this strategy and seeing these observations.  Thus all the states of
%any knowledge should have the same observations.
%We formally define the reduction below.

We assume that $\varphi \in obs$, i.e. the safety state
predicate is observable for Player I.  This is a reasonable assumption
since Player I should be able to know whether he loses the game or
not.

Consider an LTS $A=(S, s_{init} , \Sigma, \rightarrow)$.  We say that
a transition $s_1 \xrightarrow{a} s_2$ in $A$ is $obs$-visible, if the
states $s_1$ and $s_2$ have different observations (i.e.
$\gamma_{obs}(s_1) \neq \gamma_{obs}(s_2)$), otherwise we call this
transition to be $obs$-invisible.  Let $v\subseteq S$ be \emph{a knowledge} (belief) of Player I in $A$, i.e. it is some set of
states that satisfy the same observation.  The set $Post_{obs}(v, a)$
contains all the states that are accessible from the states of $v$ by
a finite sequence of $a$-labeled $obs$-invisible transitions followed
by an $a$-labeled $obs$-visible transition.  More formally,
$Post_{obs}(v, a)$ contains all the states $s'$, such that there
exists a run $s_1 \xrightarrow{a} s_2 \xrightarrow{a} \dots
\xrightarrow{a} s_n$ and $s_1 \in v$, $s_n = s'$, $\gamma_{obs}(s_i) =
\gamma_{obs}(s)$ for all $1 \leq i < n$, and $\gamma_{obs}(s_n) \neq
\gamma_{obs}(s)$.
 
The set $Post_{obs}(v, a)$ contains all the states that are visible
for Player I after he continuously offers to play action $a$ from
some state in $v$. Player I can distinguish the states $s_1$ and $s_2$
from $Post_{obs}(v, a)$ iff they have different observations,
i.e. $\gamma_{obs}(s_1) \neq \gamma_{obs}(s_2)$.  In other words, the
set $\{Post_{obs}(v, a) \cap \gamma_{obs}^{-1}(o) \mid o \in {\mathcal
  P}(obs)\} \setminus \{\varnothing\}$ consists of all the beliefs
that Player I might have after he plays the $a$ action from the
knowledge set $v$\footnote{the powerset ${\mathcal P}(S)$ is equal to
  the set of all subsets of $S$}.
%Player I can distinguish these knowledges because they have different observations.

A game can \emph{diverge} in the current observation
after playing some action.  To capture this we define the boolean
function $Sink_{obs}(v, a)$ whose value is true iff there exists an
infinite run $(s_0, s_1, \dots, s_n, \dots ) \in Runs(A)$ such that
$s_0 \in v$ and for each $i \geq 0$ we have $s_i \xrightarrow{a}
s_{i+1}$ and $\gamma_{obs}(s_i) = \gamma_{obs}(s_0)$.

%It should be noted, that the runs without observable transitions are not reflected in the set $Post_{obs}(v, a)$. 
%It follows our definition of OBSI strategies: if a game diverges (i.e. an observation is not changed), then Player I can't detect this and therefore he can't propose another action to play.

\begin{definition}
We say, that a game $(G, \psi)$ is the knowledge game for $(A, \varphi,
obs)$, if $G=(V, v_{init} , \Sigma, \rightarrow_g)$ is an LTS and
\begin{itemize}
\item $V=\{v \in {\mathcal P}(S) \mid \forall s_1, s_2 \in v \cdot \gamma_{obs}(s_1) = \gamma_{obs}(s_2)\} \setminus \{\varnothing\}$ is the set of all the beliefs of Player I in $A$,
\item $v_{init} = \{s_{init}\}$ is the initial game state,
\item $\rightarrow_g$ represents the game transition relation; a
   transition $v_1 \xrightarrow{a}_g v_2$ exists iff:
   \begin{itemize} 
    \item $v_2 = Post_{obs}(v_1, a) \cap \gamma_{obs}^{-1}(o)$ and $v_2 \neq \varnothing$ for some $o \subseteq obs$, or
    \item $Sink_{obs}(v_1, a)$ is true and $v_2 = v_1$.
   \end{itemize}
 \item $v \models \psi$ iff $\varphi \in \gamma_{obs}(v)$.
\end{itemize}
\end{definition}

\begin{theorem}[\cite{cdllr07}]
  Player I wins in a IISG $(A, \varphi, obs)$ iff he
  has a winning strategy in the safety game $(G, \psi)$ which is the knowledge game for $(A, \varphi, obs)$.
\end{theorem}

This theorem gives us the algorithm of solution of a IISG for the case when the knowledge games for it is finite and can be automatically constructed.

\section{Timed Game Automata}\label{tga}

The knowledge game $(G, \psi)$ for $(A, \varphi, obs)$ is finite when
the source game $A$ is finite~\cite{reif84}.  The converse is not true
and there are higher level formalisms that can induce \emph{infinite}
games for which knowledge games are still \emph{finite} and can be automatically constructed.  One of such
formalisms is Timed Game Automata~\cite{Maler95onthe}, that we use as
a modeling formalism and that has been proved in~\cite{cdllr07} to
have finite state knowledge games.

%This is an important result, since Timed automata is a popular formalism for the modeling of real-time systems and can handle both continuous-time and discrete behavior. 

%\marginpar{J.F., what do you propose to write here about TA?}
%\emph{Timed Game Automata is a natural game extension of Timed Automata~\cite{Alur94atheory}, and the latter is a well-known formalism for modeling real-time systems Timed Automata incorporates both discrete transitions and continuous delays in time. }

%We give a definition of Timed Game Automata below.

Let $X$ be a finite set of real-valued variables called clocks.  We
denote by ${\cal C}(X)$ the set of constraints $\psi$ generated by the
grammar: $\psi ::= x \sim k \mid x - y \sim k \mid \psi \wedge
\psi$ where $k \in \bbbn$, $x,y \in X$ and $\sim \in
\{<,\leq,=,>,\geq\}$.  ${\cal B}(X)$ is the set of constraints
generated by the following grammar: $\psi ::= \top \mid k_1 \leq x
< k_2 \mid \psi \wedge \psi$ where $k, k_1, k_2 \in \bbbn$, $k_1
< k_2$, $x \in X$, and $\top$ is the boolean constant {\em true}.

A \emph{valuation} of the clocks in $X$ is a mapping $X\mapsto
\setR_{\geq 0}$.  For $Y \subseteq X$, we denote by $v[Y]$ the
valuation assigning $0$ (respectively, $v(x)$) for any $x\in Y$
(respectively, $x\in X\setminus Y$).  We also use the notation
$\vec{0}$ for the valuation that assigns $0$ to each clock from $X$.

\begin{definition}[Timed Game Automata]\label{def-tgs} A \emph{Timed Game Automaton} (TGA) is a tuple
  $(L, l_{init} , X, E, \Sigma_c, \Sigma_u, I)$ where:
\begin{itemize}
\item $L$ is a finite set of locations,
\item $l_{init} \in L$ is the initial location,
\item $X$ is a finite set of real-valued clocks,
\item $\Sigma_c$ and $\Sigma_u$ are finite the sets of controllable and
  uncontrollable actions (of Player I and Player II, correspondingly),
\item $E \subseteq (L \times {\cal B}(X) \times \Sigma_c \times 2^X
  \times L) \cup (L \times {\cal C}(X) \times \Sigma_u \times 2^X
  \times L)$ is partitioned into controllable and uncontrollable
  transitions\footnote{We follow the definition of \cite{cdllr07} that
    also assumes that the guards of the controllable transitions
    should be of the form $k_1 \leq x < k_2$.  This allows us to use
    the results from that paper.  In particular, we use \emph{urgent}
    semantics for the controllable transitions, i.e. for any
    controllable transition there is an exact moment in time when it
    becomes enabled.},
\item $I : L \rightarrow {\cal B}(X)$ associates to each location its
  \emph{invariant}.
\end{itemize}
\end{definition}

We first briefly recall \emph{the non-game} semantics of TGA, that is
the semantics of Timed Automata (TA)~\cite{Alur94atheory}.  A state of
TA (and TGA) is a pair $(l, v)$ of a location $l \in L$ and a
valuation $v$ over the clocks in $X$.  An automaton can do two types
of transitions, that are defined by the relation $\hookrightarrow$:
\begin{itemize}
\item a {\textbf{delay}} $(l,v) \xhookrightarrow{t} (l,v')$ for some
  $t \in \setR_{> 0}$, $v'=v+t$ and $v' \models I(l)$, i.e. to stay in
  the same location while the invariant of this location is satisfied,
  and during this delay all the clocks grow with the same rate, and
 \item a  {\textbf{discrete transition}} $(l,v) \xhookrightarrow{a} (l',v')$ if there is an element $(l, g, a, Y, l') \in E$, $v \models g$ and $v' = v[Y]$, i.e. to go to another location $l'$ with resetting the clocks from $Y$, if the guard $g$ and the invariant of the target location $l'$ are satisfied.
\end{itemize}

%An automaton can do {\textbf{delay}}, i.e. stay in the same location $l$ for some time while the invariant of this location is satisfied. %, and during this stay all the clocks will be growing with the same rate.
%More formally, we say that there is a delay transition $(l,v) \xhookrightarrow{t} (l,v')$ for some $t \in \setR_{> 0}$, if $v'=v+t$ and $v' \models I(l)$.
%An automaton can also do an instant {\textbf{discrete transition}} $(l,v) \xhookrightarrow{a} (l',v')$, if there is an element $(l, g, a, Y, l') \in E$ and the clock valuation $v$ satisfies $g$ (i.e. $v \models g$). After taking this transition the clocks from $Y$ will be reset to $0$ (i.e. $v' = v[Y]$).
%We say that an action $a$ is \emph{enabled} in $s$ if there exists a discrete transition $s \xhookrightarrow{a} s'$ to some state $s'$.
%Such a transition relation $\hookrightarrow$ defines the \emph{non-game} semantics of TGA.

%\subsection{Playing with Observation Based Strategies}

In the remainder of this section, we define the \emph{game semantics}
of TGA.  As in \cite{cdllr07}, for TGA, we let observable predicates be of the
form $(K, \psi)$, where $K \subseteq L$ and $\psi \in {\cal
  B}(X)$. %, i.e. $\psi$ is induced by the following grammar $\psi ::= \top \mid k_1 \leq x < k_2 \mid \psi \wedge \psi$.
We say that a state $(l,v)$ satisfies $(K,\psi)$ iff $l \in K$ and $v
\models \psi$. 

%We start by an informal presentation and then turn to the formalization.
% Controller and Environment play on the underlying 2-LTTS of a TGS as
% follows.
%Player I has to play according to observation-based stuttering-invariant strategies.  
Intuitively, whenever the current
observation of the system state changes, Player I proposes a
controllable action $a \in \Sigma_c$ and as long as
the observation does not change Player II has to play this action when it is enabled, and otherwise he can play any uncontrollable actions or do time delay.
Player I can also propose a special action {\bf skip}, that means that he lets Player II play any uncontrollable actions and do time delay.
Any time delay should be stopped as soon as the current observation is changed, thus giving a possibility for Player I to choose another action to play.

Formally, the semantics of TGA is defined by the following definition:
\begin{definition}
The semantics of  TGA $(L, l_{init} , X, E, \Sigma_c, \Sigma_u, I)$ with the set of observable predicates $obs$ is defined as the
LTS $(S, s_{init} , \Sigma_c\cup \{{\bf skip}\}, \rightarrow)$, where  
 $S=L \times \setR_{\geq 0}^X$, $s_{init} = (l_{init}, \vec{0})$ and the transition relation is:
 ($\hookrightarrow$ denotes the non-game semantics of $M$)
\begin{itemize} 
\item  $s \xrightarrow{{\bf skip}} s'$ exists, iff $s \xhookrightarrow{a_u} s'$ for
  some $a_u \in \Sigma_u$, or  there exists a delay $s \xhookrightarrow{t} s'$ for some $t \in \setR_{> 0}$ and any smaller delay doesn't change the current observation (i.e. if $s \xhookrightarrow{t'} s''$ and $0 \leq t'<t$ then $\gamma_{obs}(s) = \gamma_{obs}(s'')$).
\item for $a\in \Sigma_c$, $s \xrightarrow{a} s'$ exists, iff:
\begin{itemize}
\item $a$ is enabled in $s$ and there exists a discrete transition $s \xhookrightarrow{a} s'$, or
\item $a$ is not enabled in $s$, but there exists a discrete transition $s \xhookrightarrow{a_u} s'$ for
  some $a_u \in \Sigma_u$, or
\item there exists a delay $s \xhookrightarrow{t} s'$ for some $t \in \setR_{> 0}$, and for
  any smaller delay $s \xhookrightarrow{t'} s''$ (where $0 \leq t'<t$) the observation is not changed, i.e. $\gamma_{obs}(s)=\gamma_{obs}(s'')$, and action $a$ is not enabled in $s''$.
\end{itemize}
\end{itemize}
\end{definition}
 
%We assume, that there exists a special controllable action ${\bf skip}
%\in \Sigma_c$, and there are no transitions labeled with this action.
%That is, if a Controller proposes to play ${\bf skip}$, he lets
%Environment play any uncontrollable action or to make any delay.
%\marginpar{Peter: it is not possible to move the definition of the ``skip'' action to the previous section since there are no controllable and uncontrollable transitions for %LTS.
%instead of it for the LTSes we have to strictly define, what is an outcome for each action}

For a given TGA $M$, set of observable predicates $obs$ and a safety
state-predicate $\varphi$ (that can be again of the form $(K, \psi)$),
we say that Player I wins in the Imperfect Information Safety Timed Game (IISTG) $(M, \varphi, obs)$ iff he
wins in the IISG $(A, \varphi, obs)$, where $A$ defines the semantics for $M$ and $obs$.

%Moreover, for each $A$ and safety property $\varphi$ there is a unique knowledge game up to isomorphism.
%Knowledge games for finite games are also finite and can be easily
%constructed by definition. %we can always construct a $(G, \psi)$.
The problem of solution of IISTG is decidable since the knowledge games are finite for
TGA~\cite{cdllr07}.
The paper \cite{cdllr07} proposes a symbolic Difference Bounded Matrices (DBM)-based procedure to construct them.  
%This results in the procedure $Solve(M, \varphi, obs)$ for solving a Partial Information Safety Timed Game $(M, \varphi, obs)$ via the reduction to a finite-state knowledge game $(A_M, \varphi, obs)$.
%This construction uses Difference Bounded Matrices (DBM) as a data structure to handle sets of clock valuations.  
%The paper \cite{cdllr07} also proposes an on-the-fly construction of $G$ being
%an extension of Liu\&Smolka
%algorithm~\cite{springerlink:10.1007/BFb0055040}, that allows us to
%stop an exploration of $G$ and backpropagation of its losing states
%early when we detect that the initial state is losing for Player I.

\section{Problem Statement}\label{problemstatement}

%The results of the previous section allow us to solve Partial 
%Information Safety Timed Games $(M, \varphi, obs)$ for a \emph{fixed}
% set of observable predicates $obs$. 
Consider that several observable predicates are available, with assigned costs, and we
look for a set of observable predicates allowing controllability
and whose cost is minimal. This is formalized in the next definition:

%Consider now that there is a set of \emph{available} observable predicates
% $Obs$, and we want to find the minimal (or the cheapest) set of observable
%  predicates $obs \subseteq Obs$ such that $(M, \varphi, obs)$ is winning for 
%  Player I. We state this problem below.

\begin{definition}
Consider a TGA $M$, a finite set of \emph{available} observable predicates 
$Obs$ over $M$, a safety observable  predicate $\varphi \in Obs$ and
a monotonic with respect to set inclusion function $\omega:{\cal P}(Obs)\rightarrow \mathbb{R}_{\geq 0}$.
The optimization problem for $(M,\varphi,Obs,\omega)$ consists in computing 
a set of observable predicates $obs \subseteq Obs$ such that
Player I wins in the Imperfect Information Safety Timed Game $(M, \varphi, obs)$
and $\omega(obs)$ is minimal.
\end{definition}

%mapping  over the states of $M$
%A set $obs\subseteq Obs$ is called a solution for the triple $(M, \varphi, Obs)$, iff Player I wins in the %Partial Information Safety Timed Game $(M, \varphi, obs)$. 
%\end{definition}

\iffalse
Consider also, that there is a real-valued function $\omega$ that \marginpar{Peter: I don't know how to fit the costs to the definition 6}
assigns costs to the possible sets of observable predicates, and
$\omega$ is monotonic, i.e. for any $obs_1, obs_2 \subseteq Obs$, if
$obs_1 \subseteq obs_2$, then $\omega(obs_1) \leq \omega(obs_2)$\footnote{in
the simplest case, a cost can be assigned to each predicate, and the
value of $\omega(obs)$ can be equal to the sum of the costs of the observable predicates from
$obs$.}.
In the
current paper we study the problem of finding a solution that minimizes
the value of the cost function $\omega$ for a given $M$, $\varphi$ and $Obs$.
\fi

We present in the next section our algorithm to compute a solution to the 
optimization problem. In this paper, we restrict our attention to \emph{finite} sets of 
available predicates. We justify this restriction by the following undecidability result:
 considering a reasonable infinite set of observation predicates, the easier problem
of the existence of a set of predicates allowing controllability is undecidable (the proof is given in Appendix~\ref{app:undec})
:

\newcounter{thm-UNDEC}
\setcounter{thm-UNDEC}{\value{theorem}}
\begin{theorem}\label{thm:undec}
Consider a TGA $M$ with clocks $X$, and an (infinite) set of available predicates
$Obs = \{x < \frac{1}q \mid x\in X, q\in \mathbb{N}, q \geq 1 \}$ and the safety objective $\varphi$. Determining whether there 
exists a finite set of predicates $obs \subset Obs$ such that Player I wins in IISTG $(M, \varphi, obs)$ is undecidable.
\end{theorem}

%\begin{proof}[Sketch]
%The proof adopts the construction of the proof of the undecidability of the existence
%of a sampling rate allowing controllability of a timed automata w.r.t. a safety objective, 
%proved in~\cite{CHR02}. Details are given in Appendix~\ref{app:undec}.
%\end{proof}

%Our algorithm consists of two ingredients.  The first is the algorithm $Solve(M, \varphi, obs)$
%for solving a game for a given (fixed) set of observable predicates $obs$.  
%This algorithm was first presented in \cite{cdllr07} and we quickly describe this algorithm in
%the section~\ref{fixed}.  Then, in section~\ref{lattice-based} we describe
%our optimization algorithm $Optimize(A, \varphi, Obs, \omega)$ for finding a solution with the minimal cost.
%This algorithm uses $Solve(M, \varphi, obs)$ as a sub-algorithm and
%avoids redundant game checks by exploiting the relations
%between possible solutions, as well as it reuses some information between the different iterations.

\section{The Algorithm}\label{lattice-based}

%Let $A_M$ be the underlying LTS for the input TGA $M$.
%Our algorithm works by iteratively solving the game $(M, \varphi, obs)$ for different observation levels $obs \subseteq Obs$ and picking a solution with the minimal s
The naive algorithm %to solve the optimization problem stated previously, 
is to iterate through all the possible solutions
${\mathcal P}(Obs)$, for each $obs \in {\mathcal P}(Obs)$ solve IISTG $(M, \varphi, obs)$ via the reduction
to the finite-state knowledge games, and finally pick a solution with
the minimal cost.

%The set of all possible solutions ${\mathcal P}(Obs)$ is finite, thus
%it is possible to find a solution with the minimal cost by solving a
%game for all the elements from ${\mathcal P}(Obs)$.

In section~\ref{basic_algorithm} we propose the more efficient algorithm that avoids exploring  all the possible solutions from ${\mathcal P}(Obs)$.
%still iteratively solves the game for different sets of observable predicates, but avoids redundant computations. % by exploiting the relations between possible solutions. 
%This basic algorithm is described in section~\ref{basic_algorithm}.
Additionally, in sections~\ref{space_reusage} we describe the optimization that reuses the information between different iterations.

\vspace{-10pt}

\subsection{Basic Exploration Algorithm}\label{basic_algorithm}

Consider, that we already solved the game for the observable
predicates sets $obs_1, obs_2, \dots, obs_n$ and obtained the results
$r_1, r_2, \dots, r_n$, where $r_i$ is either $true$ or $false$,
depending on whether Player I wins in IISTG $(M, \varphi, obs_i)$
or not.

From now on we don't have to consider any set of observable predicates
with a cost larger or equal to the cost of the optimal solution found
so far. % (since this solution will not be optimal).
Additionally, if we know, that Player I loses for the set of
observable predicates $obs_i$ (i.e. $r_i = false$), then we can
conclude that he also loses for any coarser set of observable predicates $obs \subset obs_i$ (since in this case Player I has less
observation power). Therefore we don't have to consider such $obs$ as a
solution to our optimization problem.
This can be formalized by the following definition:
\begin{definition}
  A sequence $(obs_1, r_1), (obs_2, r_2) \dots (obs_n, r_n)$ is called
  a non-redundant sequence of solutions for a set of available
  observable predicates $Obs$ and cost function $\omega$, if for any
  $1 \leq i \leq n$ we have $obs_i \subseteq Obs$, $r_i \in \{true,
  false\}$, and for any $j<i$ we have:
\begin{itemize}
 \item $\omega(obs_j) > \omega(obs_i)$ if $r_j = true$,
 \item $obs_i \not \subseteq obs_j$, otherwise.
\end{itemize}

\end{definition}

\begin{algorithm}[t]
%\SetAlgoLined \SetAlgoVlined %\DontPrintSemicolon
%\caption{Computation Of The Best Candidate For Nash Equilibrium}
\begin{small}
//input: TGA $M$, a set of observable predicates $Obs$, a safety predicate $\varphi$ \\
//output: a solution with a minimal cost \\
\textbf{function} $Optimize(M, \varphi, Obs, \omega)$: \\
1. \verb|  | $candidates := {\mathcal P}(Obs)$ // initially, $candidates$ contains all subsets of $Obs$ \\
2. \verb|  | $best\_candidate := None$  \\
3. \verb|  | \textbf{while} $candidates \neq \varnothing$: \\
4. \verb|  | \verb|  | \textbf{pick} $obs \in candidates$ \\
5. \verb|  | \verb|  | \textbf{if} $Solve(M, \varphi, obs)$: \\
6. \verb|  | \verb|  | \verb|  | $best\_candidate := obs$ \\
7. \verb|  | \verb|  | \verb|  | $candidates = candidates \setminus \{c : c \in {\mathcal P}(Obs) \land \omega(c) \geq \omega(obs)\}$ \\
8. \verb|  | \verb|  | \textbf{else}: \\
9. \verb|  | \verb|  | \verb|  | $candidates = candidates \setminus \{c : c \in {\mathcal P}(Obs) \land c \subseteq obs\}$ \\
10. \verb|  | \textbf{return} $best\_candidate$ \\
\end{small}
\caption{Lattice-based algorithm}
\label{alg:basiclatticealgorithm}
\end{algorithm}

Algorithm \ref{alg:basiclatticealgorithm} solves the optimization problem by iteratively solving the game for different sets of observable predicates so that the resulting sequence of solutions is non-redundant.
%The basic algorithm that handles these optimizations is depicted at Fig.~\ref{alg:basiclatticealgorithm}.
%The algorithms iteratively solves the game for non-redundant 
The procedure $Solve(M, \varphi, obs)$ uses the knowledge game-reduction technique described in section~\ref{lts}.
The algorithm updates the set $candidates$ after each
iteration and when the algorithm finishes, the $best\_candidate$
variable contains a reference to the solution with the minimal cost.
%Basically, in this algorithm we iteratively pick a set of observable predicates $obs$, then try to solve $(A, \varphi)$ with respect to the observations $obs$. 
%The latter is done using the procedure $Solve(A, \varphi, obs)$ that was proposed in~\cite{cdllr07} and that has been briefly described in the previous section.
%If we see that the Controller wins in the game $(A, \varphi, obs)$, then we don't have to consider set of observable predicates with cost larger than $w(obs)$ in future.
%Otherwise, we can remove from further consideration all sets of observable predicates $obs'$ such that $obs' \subset obs$, since in these games a Controller has strictly less information, %then in the game for $obs$, thus a Controller will not be able to win.

Algorithm~\ref{alg:basiclatticealgorithm} doesn't state, in which
order we should navigate through the set of candidates.  We propose
the following heuristics:
\begin{itemize}
\item \emph{cheap first} (and \emph{expensive first}) --- pick any element from the $candidates$ with the
  maximal (or minimal) cost,
\item \emph{random} --- pick a random element from the $candidates$,
\item \emph{midpoint} --- pick any element, that will allow us to
  eliminate as many elements from the $candidates$ set as it is
  possible. In other words, we pick an element that maximizes the
  value of \newline $\min(|\{c : c \in candidates \land w(c) \geq w(obs)\}|,
  |\{c : c \in candidates \land c \subseteq obs\}|)$.
\end{itemize}

Algorithm~\ref{alg:basiclatticealgorithm} doesn't specify how we
store the set of possible solutions $candidates$.  An explicit way
(i.e. store all elements) is expensive, because the $candidates$
set initially contains $2^{|Obs|}$ elements.  However, an efficient procedure 
for obtaining a next candidate may not exist as a consequence of the
following theorem that is proved in the Appendix \ref{theorem_2}:

\newcounter{thm-NP}
\setcounter{thm-NP}{\value{theorem}}
\begin{theorem}\label{np_completeness_theorem}
  Let $seq_n=(obs_1, r_1), (obs_2, r_2), \dots ,(obs_n, r_n)$ be a non-redundant sequence of solutions for some set $Obs$ and cost function $\omega : {\mathcal P}(Obs) \rightarrow \mathbb{R}_{\geq 0}$.
  Consider that the value of $\omega$ can be computed in polynomial time.  
  Then the problem of determining whether there exists a one-element extension \newline $seq_{n+1} = (obs_1, r_1), (obs_2, r_2), \dots ,(obs_n, r_n), (obs_{n+1}, r_{n+1})$ of $seq$ that is still non-redundant for $Obs$ and $\omega$ is NP-complete.
\end{theorem}

\vspace{-15pt}
\subsection{State space reusage from finer observations}\label{space_reusage}

\begin{figure}[t]
  \begin{pspicture}(10,3.4)
    \put(0.7,0){
      \put(-0.4,3.9){\textbf a)}
      \put(0.0,0.0){\includegraphics[scale=0.41]{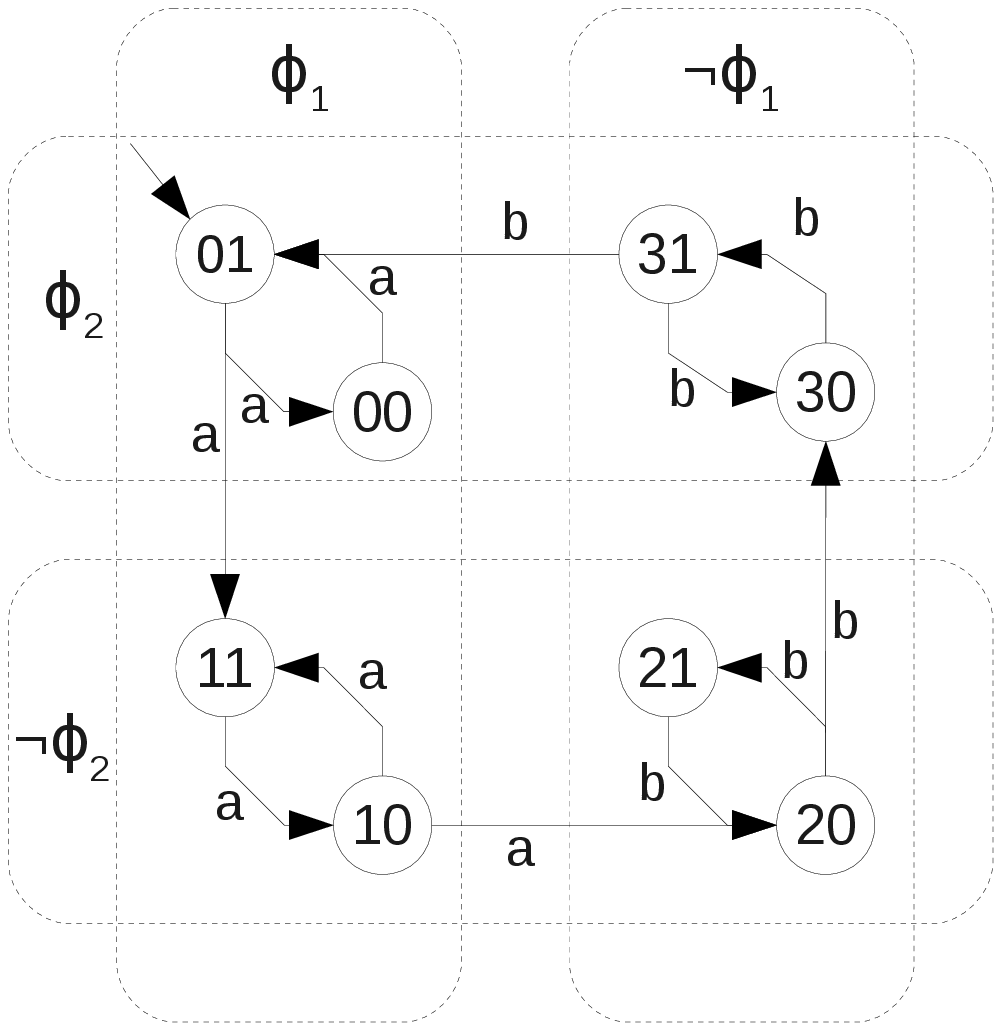} }
%      \put(0.2,3.8){\textsf{Machine}}
    }
    \put(5.5,+0.8){
      \put(-0.4,3.2){\textbf b)}
      \put(0,0){\includegraphics[scale=0.41]{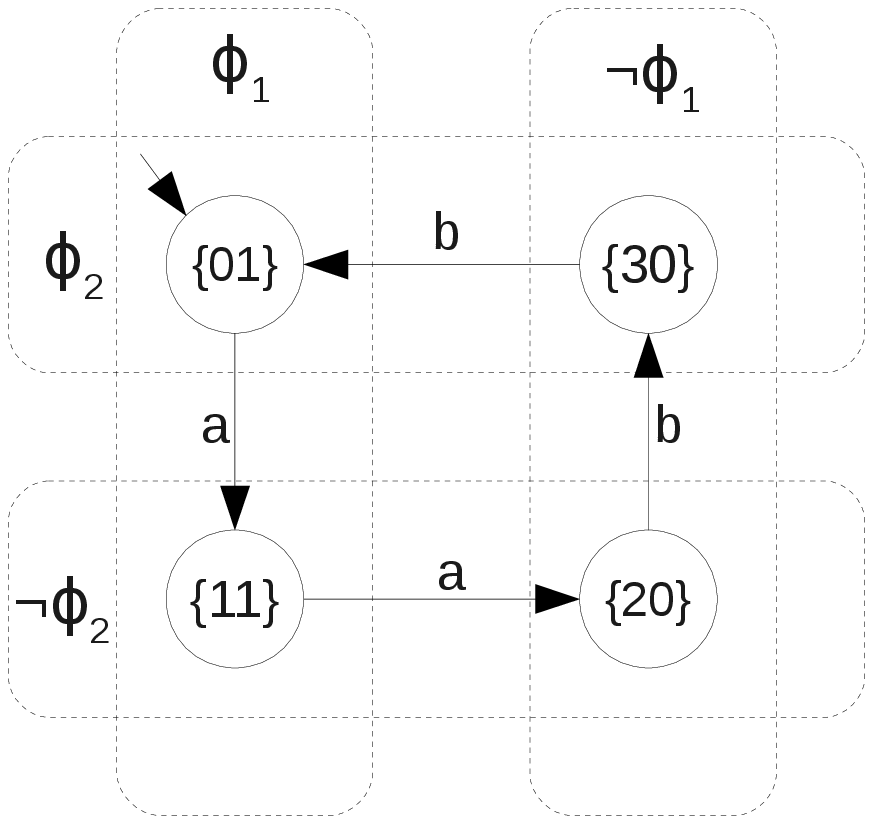}}
    }
    \put(9.8,+0.5){
      \put(-0.4,3.55){\textbf c)}
      \put(0,0){\includegraphics[scale=0.41]{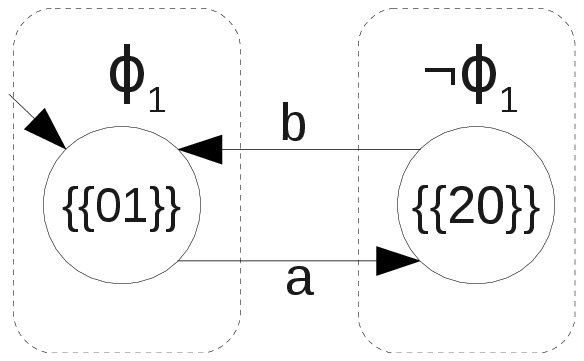}}
    }
    \put(9.8,+2.8){
      \put(-0.4,-1.10){\textbf d)}
      \put(0,0){\includegraphics[scale=0.41]{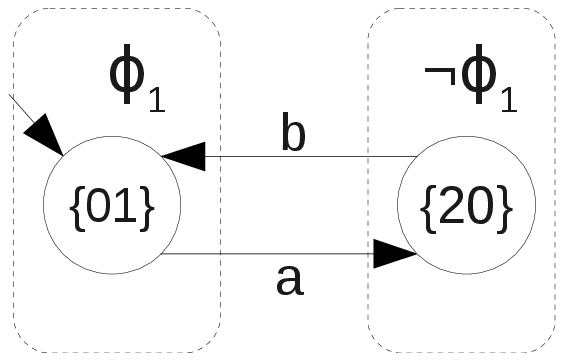}}
    }    
  \end{pspicture}
  \caption{a) The original LTS $A$ and two observable predicates $\varphi_1$ and $\varphi_2$, \newline 
     \,\, b) the knowledge game $G_{f}$ for $A$ with observable predicates $\{\varphi_1, \varphi_2\}$, \newline 
    c) the knowledge game $G^1_{c}$ for $A$ with observable predicates $\{\varphi_1\}$, \newline 
    d) the knowledge game $G^2_{c}$ for $G_{f}$ with observable predicates $\{\varphi_1\}$}
  \label{reusage:demonstration}
\end{figure}

Intuitively, if we have already solved a knowledge game $(G_{f}, \psi_f)$ for a set  $obs_f$ of observable predicates, then we can view a knowledge game $(G_{c}, \psi_c)$ associated with a \emph{coarser} set of observable predicates $obs_c \subset obs_f$ as an imperfect information game with respect to $(G_{f}, \psi_f)$. 
Thus we can solve the knowledge game for $obs$ without exploring the state space of the TGA $M$ and therefore without using the expensive DBM operations.
Moreover, we can build another game on top of $G_{c}$ (for an observable predicates set that is coarser than $obs$) and thus construct a ``Russian nesting doll'' of games.
This is an important contribution of our paper, since this construction can be applied not only to Timed Games, but also to any modeling formalism that have finite knowledge games.

The state space reusage method is demonstrated on a simple LTS $A$ at Fig.~\ref{reusage:demonstration}.
%(for the sake of simplicity we don't define safety predicates).
Suppose, that we already built the knowledge game $G_{f}$ for the observable predicates $\{\varphi_1, \varphi_2\}$.
Now, if we want to build a knowledge game for $\{\varphi_1\}$, we can do that in two ways.
First, we can build it from scratch based on the state space of $A$, and the resulting knowledge game $G^1_{c}$ is given at subfigure c.
Alternatively, we can build the knowledge game $G^2_{c}$ on the top of $G_{f}$ (see subfigure d).
The states of $G^1_{c}$ are sets of states of $A$ and the states of $G^2_{c}$ are sets of sets of states of $A$.
The games $G^2_{c}$ and $G^1_{c}$ are bisimilar, thus Player I wins in $G^1_{c}$ iff he wins in $G^2_{c}$ (for any safety predicate).
The latter is true for any LTS $A$, that is stated by the following theorem and corollary (that are proved in the appendix \ref{theorem_3})
:

\newcounter{thm-REUSE}
\setcounter{thm-REUSE}{\value{theorem}}
\begin{theorem}\label{reusage_theorem}
 Suppose that $obs_c \subset obs_f$, $(G_{f}, \psi_{f})$ is the knowledge game for $(A, \varphi, obs_f)$,
 $(G^1_{c}, \psi^1_{c})$ is the knowledge game for $(A, \varphi, obs_c)$ and $(G^2_{c}, \psi^2_{c})$ is the knowledge game for $(G_{f}, \psi_{f}, obs_c)$. 
 %Then the state $v'$ is winning for Player I in $(G^2_{c}, \psi^2_{c})$ iff the state $\bigcup_{s' \in v'}s'$ is winning in for Player I in $(G^1_{c}, \psi)$.
 Then the relation $R = \{(v, v') | v = \bigcup_{s' \in v'}s'\}$ between the states of $G^1_{c}$ and $G^2_{c}$ is a bisimulation.
\end{theorem}

\begin{corollary} 
  Player I wins in $(G^1_{c}, \psi^1_{c})$ iff Player I wins in $(G^2_{c}, \psi^2_{c})$.
\end{corollary}

%Since the games $G^1_{c}$ and $G^2_{c}$ are bisimilar, and the bisimulation relation preserves the winning conditions, $(G^1_{c}, \psi_{c})$ is winning for Player I iff $(G^2_{c}, \psi^2_{c})$ is winning for Player I. 

%The consequence of this theorem is that Player I wins in $(G^2_{c}, \psi')$ iff this player wins in $(G^1_{c}, \psi)$.
This reusage method is also correct for the case when an input model is defined as a TGA (since we can apply the theorem to the underlying LTS).
%Therefore the method allows us to avoid an expensive DBM-based construction when we have to solve the knowledge game for the set $obs$ of observable predicates and we have already solved it for the finer set of observable predicates $obs'$.
%Moreover, we can build another game on top of $G^2_{c}$ (for an observable predicates set that is coarser than $obs$) and thus construct a ``Russian nesting doll'' of games.

%\subsection{Implementation} 

%We implemeted this algorithm in a prototype tool written in the Python language.
%The implementation uses the  Timed Automata models parser \verb|pyuppaal| being a part of our Timed Automata manipulation framework written in Python\footnote{\url{http://people.cs.aau.dk/\textasciitilde pbulychev/python.html}}, and the {\sc Uppaal DBM} library\footnote{\url{http://people.cs.aau.dk/ \textasciitilde adavid/UDBM}}.

%Our implementation explicitly stores the set of possible solutions $candidates$ (cf. Algorithm \ref{alg:basiclatticealgorithm}), and according to the theorem~\ref{np_completeness_theorem} there doesn't exist asymptotically better method for doing that.
%The implementation uses the on-the-fly algorithm for solving two-player games, that stops early when it detects that the initial state is losing.
%Therefore the implementation doesn't support the winning states reusage technique described in section~\ref{winning_reusage}.
%The tool supports the state space reusage tenchinque described in section~\ref{space_reusage}.

\paragraph{\textbf{Implementation}}
Our Python prototype implementation of this algorithm (see \url{https://launchpad.net/pytigaminobs})
%We developed a prototype implementation of this algorithm in Python programming language (see \url{https://launchpad.net/pytigaminobs}).
explicitly stores the set of candidates and uses the on-the-fly DBM-based algorithm of \cite{cdllr07} for the construction and solution of knowledge games for IISTG (the algorithm stops early when it detects that the initial state is losing).

\vspace{-5pt}
\section{Case studies}\label{casestudies}
\vspace{-10pt}
We applied our implementation to two case studies. 

The first is a ``Train-Gate Control'', where two trains tracks
merge together on a bridge and the goal of the controller is to
prevent their collision. % of the trains on the bridge. , and prevent the trains from staying too long at the semaphore.  
The trains can arrive
in any order (or don't arrive at all), thus the challenge for the
controller is to handle all possible cases.
%There is much uncertainty for the controller for this model and thus the performance of our tool does not scale well when the model size grows.
%Thus we apply our tool to only one relatively small instance of the model.

The second is ``Light and Heavy boxes'', where a box is being processed on the conveyor in several steps, and the goal of the controller is to move the box to the next step within some time bound after it has been processed at the current step.
%The game size grows linearly when we scale the model (increase the number of steps), thus the performance of our tool scales well. 

\begin{figure}
%\begin{tabular}{c}
\center{
\includegraphics[width=0.8\textwidth]{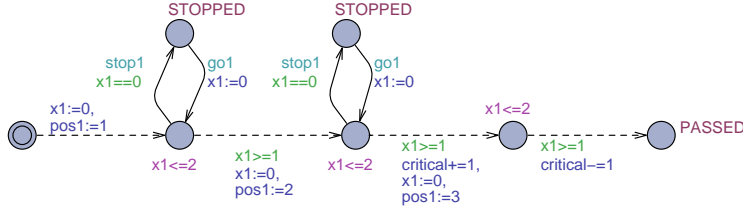} %&
}
%\raisebox{+0.0\height}{\includegraphics[width=0.13\textwidth]{clock}}
%\end{tabular}
\caption{A model of a single train}%(left) and a discritized timer(right)}
\label{traingate}
\end{figure}
\vspace{-10pt}
\subsection{Train-Gate control}
The model of a single (first) train is depicted at Fig.~\ref{traingate}.
There are two semaphore lights before the bridge on each track.  A
train passes the distance between semaphores within $1$ to $2$ time units.
A controller can switch the semaphores to red (actions \verb|stop1|
and \verb|stop2| depending on the track number), and to green (actions
\verb|go1| and \verb|go2|).  These semaphores are intended to prevent
the trains from colliding on the bridge.  When the red signal is
illuminated, a train will stop at the next semaphore and wait for the
green signal.

It is possible to mount sensors on the semaphores, and these sensors
will detect if a train approaches the semaphore.  This is modeled with
observable predicates $(pos1 \geq 1)$, $(pos2 \geq 1)$, $(pos1 \geq
2)$ and $(pos2 \geq 2)$.

\begin{figure}
\begin{center}

\subfloat[Running time (the average is computer on $10$ runs)] {
\label{tab:1}
\begin{tabular}{|l|c|c|c|c|c|c|c|c|}
   \hline
   exploration order            & \multicolumn{2}{|c|}{expensive first} & \multicolumn{2}{|c|}{cheap first} & \multicolumn{2}{|c|}{midpoint} & \multicolumn{2}{|c|}{random} \\ \hline
   state space reusage          & with             & without            &       with      &     without     &    with     &    without       &   with     &    without      \\ \hline
   minimum                  & 10m              & 1h03m              & 50m             & 49m             & 24m         &   41m            & 10m        &  48m          \\ 
    \hline
   maximum      & 11m              & 1h36m              & 1h30m           & 1h34m           & 55m         &   1h36m          & 1h26m      &  1h44m          \\ \hline
   average & 10m              & 1h18m              & 1h0m            & 1h12m           & 33m         &   1h03m          & 37m        &  1h05m          \\ \hline   
 
\end{tabular}
}

\subfloat[The average number of iterations] {
\label{tab:2}
\begin{tabular}{|l|c|c|c|c|c|c|}
   \hline
   exploration order            & expensive first & cheap first & midpoint & random \\ \hline
   without state space reusage  & 1                  & 21.69           & 5.27           & 6.17         \\
    \hline
   with state space reusage & 7.1 & 0 & 2.7 & 3.46

 \\ \hline
 
\end{tabular}
}

\end{center}
\caption{Results for the Train-Gate model}
\label{traingate_result}
\end{figure}

The controller has a discrete timer that is modeled using the clock
$y$.  At any time this clock can be reset by the controller (action
\verb|reset|).  There is an available observable predicate $(y < 2)$
that becomes false when the value of $y$ reaches $2$.  This allows the
controller to measure time with a precision $2$ by resetting $y$ each
time this predicate becomes false and counting the number of such resets.

The integer variable $critical$ contains the number of trains that are
currently on the bridge.  The safety property is that no more than one
train can be at the critical section (bridge) at the same time and the
trains should not be stopped for more than $2$ time units:
\vspace{-5pt}
$$
(critical < 2) \land ((Train1.STOPPED) \rightarrow (x1 \leq 2)) \land ((Train2.STOPPED) \rightarrow (x2 \leq 2)) 
$$
%Additionally, we don't want a controller to permanently stop all the trains.
%Thus we require that the trains should not be stopped for more than $2$ time units.
\vspace{-5pt}

The optimal controller uses the following set of observable predicates: $(pos1 \geq 2)$, $(pos2
\geq 2)$ and $(y < 2)$.  Such a controller waits until the second (in
time) train comes to the second semaphore, then pauses this train and
lets it go after $2$ time units.  

Figure~\ref{tab:1} reports the time needed to find this solution for different parameters of the algorithm.
Figure~\ref{tab:2} contains the average number of iterations of
Algorithm~\ref{alg:basiclatticealgorithm} (i.e. game checks for different sets of observable predicates).
You can see that it requires only a fraction of the total number of all
possible solutions $2^5 = 32$.  Additionally, the state space reusage heuristic allows to
improve the performance, especially for the ``expensive first''
exploration order. % (in this case we reuse the state space for each iteration after the first one).
For this model the most efficient way to solve the optimization problem is to first solve the game with all the available predicates being observed, and then always reuse the state space of this knowledge game.
The numbers of $0$ and $1$ at Figure~\ref{tab:2} reflect that we \emph{don't} reuse the state space exactly once for the ``expensive first'' order, and we never reuse the state space for the ``cheap first'' exploration order.
%When the ``expensive first'' exploration order is used, we solve the game without state space reusage only once (thus the average number $1$), 
%when the ``cheap first'' order is used, the state space is never reused (thus the average number $0$).

The game size ranges from $5$ states for the game when only the safety state predicate is observable to $9202$ for the case when all the available predicates are observable.
The number of the symbolic states of TGA (i.e. different pairs of reachable locations and DBMs that form the states of a knowledge game) ranges from $1297$ to $31171$, correspondingly.

\subsection{Light and Heavy Boxes}
\vspace{-10pt}

\begin{figure}
\begin{center}
\includegraphics[width=0.5\textwidth]{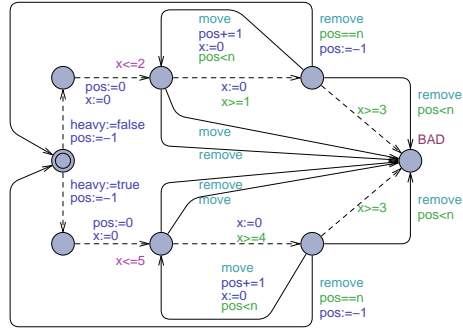}
\end{center}
\vspace{-5pt}
\caption{Light and heavy boxes model}
\label{lhboxes}
\end{figure}

\begin{figure}
\begin{center}
\includegraphics[width=1\textwidth]{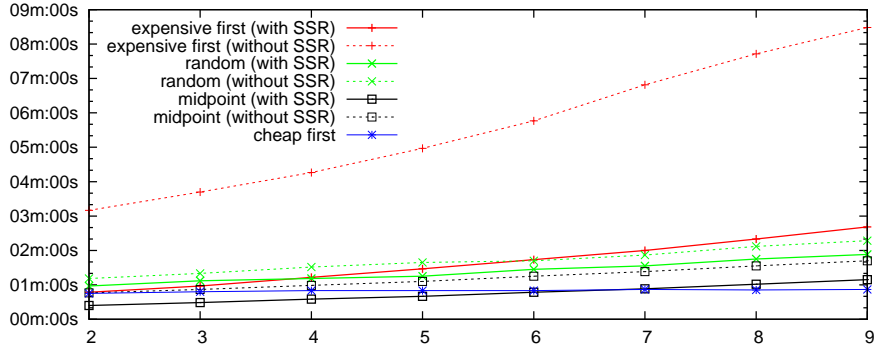}
\end{center}
\caption{Average running time (SSR states for State Space Reusage)}
\label{lhboxes_results}
\end{figure}

\vspace{-10pt}
Consider a conveyor belt on which Light and Heavy boxes can be put.
A box is processed in $n$ steps ($n$ is a parameter of the model), and the processing at each step takes from $1$ to $2$ time units for the Light boxes, and from $4$ to $5$ time units for the Heavy boxes.
The goal of the controller is to move a box to the next step (by rotating the conveyor, with an action \verb|move|) within $3$ time units after the box has been processed at the current step. 
At the last step the controller should remove (action \verb|remove|) the box from the conveyor within $3$ time units. If the controller rotates the conveyor too early (before the box has been processed), too late (after more than $3$ time units), or does not move it at all, then the Controller loses (similar is true for the removing of the box at the last step). Additionally, the controller should not rotate the conveyor when there is no box on it, and should not try to remove the box when the box is not at the last step.
Our model is depicted at Fig.~\ref{lhboxes}, and the goal of the controller is to avoid the \verb|BAD| location.

A box can arrive on the conveyor at any time, and there is an observable predicate $(pos=0)$ with cost $1$ which becomes true when the box is put on the conveyor. Additionally, there is predicate $(heavy=true)$ with cost $1$ that becomes true if a heavy box arrives.
The model is cyclic, i.e. another box can be put on the conveyor after the previous box has been removed from it.

As in the Traingate model, the controller can measure time using a special clock $y$.
We assume that a controller can measure time with different granularity, and more precise clocks cost more.
We model this by having three available observable predicates: $(y < 1)$ with cost $3$, $(y < 2)$ with cost $2$, and $(y < 3)$ with cost $1$.

A naive controller works with the observable predicates $\{(heavy=true), (pos=0), (y<1)\}$, resets the clock $y$ each time a new box is arrived, and then move it to the next step (remove after the last iteration) each $2$ time units if the box is \emph{light} and $5$ time units if the box is \emph{heavy}.
However, it is not necessary to use the expensive $(y<1)$ observable predicate, since a controller can move a box after each $3$ ($6$ for heavy box) time units, thus the time granularity of $3$ is enough and there is a controller that uses the observable predicates $\{(heavy=true), (pos=0), (y<3)\}$.
Our implementation detects such an optimal solution, and Fig.~\ref{lhboxes_results} demonstrates an average time needed to compute this solution for different numbers of box processing steps $n$.
%We used $10$ experiments for the estimating the midpoint, most expensive and most cheap search times, and $100$ experiments for the random exploration (since it's more biased).
You can see that the state space reusage heuristics improves the performance of the algorithm.

The game size for this model ranges from $4$ knowledge game states and $51$ symbolic NTA states when there are $2$ processing steps and only safety predicate is observable to $6417$ knowledge game states and $15554$ symbolic NTA states for $9$ processing steps and when all the available predicate are observable.

\section{Conclusions}

In this paper we have developed, implemented and evaluated an algorithm for the cost-optimal controller synthesis for timed systems, where the cost of a controller is defined by its observation power.

Our important contributions are two optimizations: the one that helps to avoid exploration of all possible solutions and the one that allows to reuse the state space and solve the imperfect information games on top of each other.
Our experiments showed that these optimizations allow to improve the performance of the algorithm.

In the future, we plan to apply our method to other modeling formalisms that have finite state knowledge games.

\bibliographystyle{plain}
\bibliography{main}

\begin{thebibliography}{10}

\bibitem{guided_controller_synthesis}
Israa AlAttili, Fred Houben, Georgeta Igna, Steffen Michels, Feng Zhu, and
  Frits~W. Vaandrager.
\newblock Adaptive scheduling of data paths using uppaal tiga.
\newblock In {\em QFM}, pages 1--11, 2009.

\bibitem{Alur94atheory}
Rajeev Alur and David~L. Dill.
\newblock A theory of timed automata.
\newblock {\em Theoretical Computer Science}, 126:183--235, 1994.

\bibitem{bcdfll07}
Gerd Behrmann, Agnes Cougnard, Alexandre David, Emmanuel Fleury, Kim~G. Larsen,
  and Didier Lime.
\newblock {\sc Uppaal-tiga}: Time for playing games!
\newblock In {\em Proceedings of the 19th International Conference on Computer
  Aided Verification}, number 4590 in LNCS, pages 121--125. Springer, 2007.

\bibitem{journals/fmsd/BouyerBBR07}
Patricia Bouyer, Thomas Brihaye, Véronique Bruyère, and Jean-François
  Raskin.
\newblock On the optimal reachability problem of weighted timed automata.
\newblock {\em Formal Methods in System Design}, 31(2):135--175, 2007.

\bibitem{DBLPconffsttcsBouyerCFL04}
Patricia Bouyer, Franck Cassez, Emmanuel Fleury, and Kim~Guldstrand Larsen.
\newblock Optimal strategies in priced timed game automata.
\newblock In {\em FSTTCS}, pages 148--160, 2004.

\bibitem{BouyerDMP03}
Patricia Bouyer, Deepak D'Souza, P.~Madhusudan, and Antoine Petit.
\newblock Timed control with partial observability.
\newblock In {\em CAV}, volume 2725 of {\em Lecture Notes in Computer Science},
  pages 180--192. Springer, 2003.

\bibitem{BrihayeBR05}
Thomas Brihaye, V{\'e}ronique Bruy{\`e}re, and Jean-Fran\c{c}ois Raskin.
\newblock On optimal timed strategies.
\newblock In {\em FORMATS}, volume 3829 of {\em Lecture Notes in Computer
  Science}, pages 49--64. Springer, 2005.

\bibitem{cdllr07}
F.~Cassez, A.~David, K.~G. Larsen, D.~Lime, and J.-F. Raskin.
\newblock Timed control with observation based and stuttering invariant
  strategies.
\newblock In {\em Proceedings of the 5th International Symposium on Automated
  Technology for Verification and Analysis}, volume 4762 of {\em LNCS}, pages
  192--206. Springer, 2007.

\bibitem{cdfll05}
Franck Cassez, Alexandre David, Emmanuel Fleury, Kim~G. Larsen, and Didier
  Lime.
\newblock Efficient on-the-fly algorithms for the analysis of timed games.
\newblock In {\em CONCUR'05}, volume 3653 of {\em LNCS}, pages 66--80.
  Springer--Verlag, August 2005.

\bibitem{CHR02}
Franck Cassez, {\relax Th}omas~A. Henzinger, and Jean-Fran{\c c}ois Raskin.
\newblock A comparison of control problems for timed and hybrid systems.
\newblock In {\em Proc. 5th International Workshop on Hybrid Systems:
  Computation and Control (HSCC'02)}, volume 2289 of {\em Lecture Notes in
  Computer Science}, pages 134--148. Springer, 2002.

\bibitem{automatic_synthesis}
Franck Cassez, Jan~J. Jessen, Kim~G. Larsen, Jean-Fran\c{c}ois Raskin, and
  Pierre-Alain Reynier.
\newblock Automatic synthesis of robust and optimal controllers --- an
  industrial case study.
\newblock In {\em Proceedings of the 12th International Conference on Hybrid
  Systems: Computation and Control}, HSCC '09, pages 90--104, Berlin,
  Heidelberg, 2009. Springer-Verlag.

\bibitem{flexible_timeline_based}
A.~Cesta, A.~Finzi, S.~Fratini, A.~Orlandini, and E.~Tronci.
\newblock Flexible timeline-based plan verification.
\newblock In Bärbel Mertsching, Marcus Hund, and Zaheer Aziz, editors, {\em KI
  2009: Advances in Artificial Intelligence}, volume 5803 of {\em Lecture Notes
  in Computer Science}, pages 49--56. Springer Berlin / Heidelberg, 2009.

\bibitem{ChatterjeeHJS11}
Krishnendu Chatterjee, Thomas~A. Henzinger, Barbara Jobstmann, and Rohit~Singh
  0002.
\newblock Quasy: Quantitative synthesis tool.
\newblock In Parosh~Aziz Abdulla and K.~Rustan~M. Leino, editors, {\em TACAS},
  volume 6605 of {\em Lecture Notes in Computer Science}, pages 267--271.
  Springer, 2011.

\bibitem{ChatterjeeMH08}
Krishnendu Chatterjee, Rupak Majumdar, and Thomas~A. Henzinger.
\newblock Controller synthesis with budget constraints.
\newblock In {\em HSCC}, volume 4981 of {\em Lecture Notes in Computer
  Science}, pages 72--86. Springer, 2008.

\bibitem{DimitrovaF08}
Rayna Dimitrova and Bernd Finkbeiner.
\newblock Abstraction refinement for games with incomplete information.
\newblock In {\em FSTTCS}, volume~2 of {\em LIPIcs}, pages 175--186. Schloss
  Dagstuhl - Leibniz-Zentrum fuer Informatik, 2008.

\bibitem{HenzingerJM03}
Thomas~A. Henzinger, Ranjit Jhala, and Rupak Majumdar.
\newblock Counterexample-guided control.
\newblock In {\em ICALP}, volume 2719 of {\em Lecture Notes in Computer
  Science}, pages 886--902. Springer, 2003.

\bibitem{DBLPconfhybrid2009}
Rupak Majumdar and Paulo Tabuada, editors.
\newblock {\em Hybrid Systems: Computation and Control, 12th International
  Conference, HSCC 2009, San Francisco, CA, USA, April 13-15, 2009.
  Proceedings}, volume 5469 of {\em Lecture Notes in Computer Science}.
  Springer, 2009.

\bibitem{Maler95onthe}
Oded Maler, Amir Pnueli, and Joseph Sifakis.
\newblock On the synthesis of discrete controllers for timed systems.
\newblock In {\em in E.W. Mayr and C. Puech (Eds), Proc. STACS'95, LNCS 900},
  pages 229--242. Springer, 1995.

\bibitem{MalerPS95}
Oded Maler, Amir Pnueli, and Joseph Sifakis.
\newblock On the synthesis of discrete controllers for timed systems (an
  extended abstract).
\newblock In {\em STACS}, pages 229--242, 1995.

\bibitem{MNR-ATVA11}
Janusz Malinowski, Peter Niebert, and Pierre-Alain Reynier.
\newblock A hierarchical approach for the synthesis of stabilizing controllers
  for hybrid systems.
\newblock In {\em Proc. 9th International Symposium on Automated Technology for
  Verification and Analysis (ATVA'11)}, volume 6996 of {\em Lecture Notes in
  Computer Science}, pages 198--212. Springer, 2011.

\bibitem{tga_based_controllers}
Andrea Orlandini, Alberto Finzi, Amedeo Cesta, and Simone Fratini.
\newblock {TGA}-based controllers for flexible plan execution.
\newblock In Joscha Bach and Stefan Edelkamp, editors, {\em KI 2011: Advances
  in Artificial Intelligence}, volume 7006 of {\em Lecture Notes in Computer
  Science}, pages 233--245. Springer Berlin / Heidelberg, 2011.

\bibitem{RaskinCDH07}
Jean-Fran\c{c}ois Raskin, Krishnendu Chatterjee, Laurent Doyen, and Thomas~A.
  Henzinger.
\newblock Algorithms for omega-regular games with imperfect information.
\newblock {\em Logical Methods in Computer Science}, 3(3), 2007.

\bibitem{reif84}
John~H. Reif.
\newblock The complexity of two-player games of incomplete information.
\newblock {\em Journal of Computer and System Sciences}, 29(2):274--301,
  October 1984.

\bibitem{Zwick96thecomplexity}
Uri Zwick and Mike Paterson.
\newblock The complexity of mean payoff games on graphs.
\newblock {\em Theoretical Computer Science}, 158:343--359, 1996.

\end{thebibliography}

\appendix

\section{Proof of the theorem~\ref{thm:undec}}
\label{app:undec}

\setcounter{theorem}{\value{thm-UNDEC}}
\begin{theorem} 
Consider a TGA $M$ with clocks $X$, and an (infinite) set of available predicates
$Obs = \{x < \frac{1}q \mid x\in X, q\in \mathbb{N}, q \geq 1 \}$ and the safety objective $\varphi$. Determining whether there 
exists a finite set of predicates $obs \subset Obs$ such that Player I wins in IISTG $(M, \varphi, obs)$ is undecidable.
\end{theorem}

\begin{proof}
The proof below adopts the construction of the proof of the undecidability of the existence
of a sampling rate allowing controllability of a timed automata w.r.t. a safety objective, 
proved in~\cite{CHR02}.

  The proof is by reduction to boundedness of $2$-counter machines.
 It is based on the encoding of such a machine into a timed 
 automaton $M$ described in \cite{CHR02}. We do not recall this construction here. Formally, we write 
 $M=(L, l_{init} , X, E, \Sigma_c, \Sigma_u, I)$, and recall that
 the construction involves a location  $\textsf{over}$.
 The
  main property of this construction used in this proof is the
  following: given a rational number $\frac{1}{b}\in \mathbb{Q}_{>0}$
  and denoting $k=\lfloor b \rfloor$, the counters of $M$ never exceed
  value $k$ if and only if location $\textsf{over}$ is not reachable
  in the semantics of $\cal A$ sampled by $\frac{1}b$.
 In addition, we
  slightly modify the construction as follows: we add a fresh clock
  $z$ which, along every transition, is systematically reset and
  checked to be positive. Note that this has no effect on the sampled
  semantics of the automaton, whatever the value of the sampling. 
  We also consider that every transition are controllable. 
 
  We consider as the safety objective the set 
 $(L\setminus\{\textsf{over}\},\mathbb{R}_{\geq 0}^X)$. 
  The set of observable predicates
  is defined as $Obs_1\biguplus Obs_2$, where $Obs_1$ is the set of predicates $(\ell,\mathbb{R}_{\geq 0}^X)$, 
  for every location $\ell \in L$, and $Obs_2$ is the set of predicates $(L,z<\frac{1}{q})$, 
  where  $q\in \mathbb{N}^*$.
   We will show that there exists a finite set of predicates for
  which the system is controllable if and only if the $2$-counter
  machine is bounded. 
   
  Assume the machine is bounded, say by value $k$. Thanks to the property of 
   $M$ recalled above, the semantics of $M$ for the sampling rate 
   $\frac1k$ never enters location $\textsf{over}$, and thus verifies the 
   safety objective. 
   We prove that the system is controllable for the (finite) set of predicates
   $Obs_1 \cup \{(L,z<\frac1k)\}$.
  The strategy of the controller is as follows: it alternates between delays
  (action {\bf skip}) and  discrete actions. As clock $z$ is reset on every transition, 
  this allows to simulate a sampled behavior, for sampling rate $\frac1k$.
  Indeed, after each discrete step, the value of clock $z$ is zero, and thus
  the predicate $(L,z<\frac1k)$ will become false exactly after $\frac1k$ time 
  units. At this time, the controller proposes an action $\sigma$ which is enabled.
  The outcome of this strategy is a run of $M$ under the sampled semantics,
  for the sampling rate $\frac1k$. In particular, this run never enters location
  $\sf over$.

  Conversely, we proceed by contradiction and assume that the machine
  is not bounded and that there exists a finite subset $obs$ of 
  $Obs_1\biguplus Obs_2$ which allows to control the system.
  Let us denote by $k$ the least common multiple of integers $q$ such that
  the predicate $(L,z<\frac1q)$ belongs to $obs$.
  Since we have modified the machine by
  requesting some positive delay to elapse between two discrete
  actions, and by the semantics of stuttering invariant strategies, we 
  know that all the actions of the controller will be
  played on ticks of sampling $\frac{1}{k}$ (but not necessarily all
  of them, the controller could propose to wait, or could use a predicate which is a multiple of $\frac{1}{k}$). As a consequence,
  the controlled behavior of $M$ will be a subset of the sampling 
  semantics of $M$
  w.r.t. sampling unit $\frac{1}{k}$. This constitutes a contradiction
  since, as the machine is not bounded and by properties of the
  construction of $M$, any sampled behavior of $M$ will
  eventually either be blocked or reach location $\textsf{over}$.
\qed

\iffalse 
   consider the
  predicate $z<\frac1k$. As clock $z$ is reset on every transition, we 
   can obtain a sampled behavior by Note that to produce a controller with a
  sampled behavior, the controller simply alternates between delay
  phases and action phases. In a delay step, it proposes the $\tau$
  action as long as the history clock referring to the last observation
  has not reached value $\frac{1}{k}$. In an action step, it should
  propose an action $\sigma$ which is enabled. Then the controller,
  which is sampled by the smallest rational value $\frac{1}{k}$, and
  which simply executes the machine, is correct. Indeed, otherwise
  there exists an execution in the sampling semantics of $\cal A$
  which reaches the location $\textsf{over}$. This is equivalent to
  the existence of an execution of the machine $M$ along which one of
  the counter values exceeds $k$, by construction of $\cal A$. Note
  that the environment can block time elapsing by playing a Zeno
  sequence but though this could be a problem w.r.t. a reachability
  objective, this still constitutes a correct behavior w.r.t. a
  safety objective.
\fi
\end{proof}

\section{Proof of the theorem~\ref{np_completeness_theorem}}
\label{theorem_2}

\setcounter{theorem}{\value{thm-NP}}
\begin{theorem} 
  Let $seq_n=(obs_1, r_1), (obs_2, r_2), \dots ,(obs_n, r_n)$ be a non-redundant sequence of solutions for some set $Obs$ and cost function $\omega : {\mathcal P}(Obs) \rightarrow \mathbb{R}_{\geq 0}$.
  Consider that the value of $\omega$ can be computed in polynomial time.  
  Then the problem of determining whether there exists a one-element extension \newline $seq_{n+1} = (obs_1, r_1), (obs_2, r_2), \dots ,(obs_n, r_n), (obs_{n+1}, r_{n+1})$ of $seq$ that is still non-redundant for $Obs$ and $\omega$ is NP-complete.
\end{theorem}
\begin{proof}

  First, we show that the problem is in NP.  Indeed, a proof certificate is a sequence $seq_{n+1}$ itself, and its fitness can be checked in polynomial time.
 
 Next, we demonstrate the NP-hardness by showing the reduction from the vertex cover problem.
 Formally, a vertex cover of a graph $G=(V, E)$ is a set $C \subseteq V$ of vertices such that each edge of $G$ is incident to at least one vertex in $C$.
 The vertex cover \emph{decision problem} is to determine for a given $G$ and $k$, whether there exists a vertex cover $C$ of the graph $G$, and the size of the set $C$ should be at most $k$.
  This problem is known to be NP-complete.
%We consider the case a cost function $\omega$ is additive, i.e. $\omega(obs) = \sum_{\varphi \in obs}\omega(\{\varphi\})$ for any $obs \subseteq Obs$ (and thus the theorem holds for the general form of $\omega$). 
  
 Consider, that we are given a graph $(V, E)$ and a constant $k$, and we want to check if there exists a vertex cover of size at most $k$.
 Consider also that $|V| = m$, $|E| = n$, and $E=\{e_1, e_2, \dots, e_n\}$.
 Let's choose the set $Obs$ to be equal to $V \cup \{ o_c \}$, where $o_c \not \in V$ is a special element.
 Let's define the value of the cost function $\omega(obs)$ to be equal to $|obs|+k$ if $o_c \in obs$ and to be equal to $|obs|$, otherwise.
 Consider a set ${O} = \{obs_1, obs_2, \dots, obs_n\}$ of subsets of $Obs$, where $|{O}| = |E|$, and $obs_i$ contains all the vertices from $V$ that are not incident to the edge $e_i$.
 Let's order the elements from ${O}$ in a sequence $(obs_{a_1}, obs_{a_2}, \dots, obs_{a_n})$ such that $|obs_{a_i}| \leq |obs_{a_j}|$ for any $i < j$.
 All the sets of this sequence are pairwise different (i.e. $obs_{a_i} \neq obs_{a_j}$ for $i \neq j$), and thus  $obs_{a_i} \not \subseteq obs_{a_j}$ for any $i < j$.
 %We can extract from $O$ a subset of pairwise distinct observations, and we can order them in a sequence $(obs_{a_1}, obs_{a_2}, \dots, obs_{a_k})$ so that $obs_{a_i} \not \subseteq obs_{a_j}$ for any $i > j$.
 
 Let us consider a sequence $$seq_n = (obs_{a_1}, false),(obs_{a_2}, false) \dots, (obs_{a_n}, false), (\{o_c\}, true)$$
 First, it can be easily seen that this sequence is non-redundant for the chosen $Obs$ and $\omega$ (since $obs_{a_i} \not \subseteq obs_{a_j}$ for any $i < j$).
 
 And second, the sequence $seq_n$ can be extended by one element such that the resulting sequence is still non-redundant iff there exists a vertex cover of size at most $k$ for the graph $(V, E)$.
 Indeed, such an extension exists iff there exists a set $obs \subseteq Obs$ such, that $\omega(obs) < \omega(\{o_c\})$, and $obs \not \subseteq obs_i$ for all $i=1..n$.
 Note that $\omega$ is monotonic and therefore $\omega(obs) < \omega(\{o_c\})$ implies that $o_c \not \in obs$, and thus $\omega(obs) = |obs|$.
 Therefore, we should have $|obs| < k+1$ (since $\omega(\{o_c\})= k+1$) and for any $i$ there should be some vertex $v$ that belongs to $obs$ and doesn't belong to $obs_i$.
 The latter is equivalent to the fact that there exists a set of vertices $obs$ of size at most $k$ and any edge from $E$ is incident to at least one vertex in $obs$.
 This proves the NP-hardness. \qed
\end{proof}

\section{Proof of the theorem \ref{reusage_theorem}}
\label{theorem_3}

\setcounter{theorem}{\value{thm-REUSE}}
\begin{theorem}
 Suppose that $obs_c \subset obs_f$, $(G_{f}, \psi_f)$ is the knowledge game for $(A, \varphi, obs_f)$,
 $(G^1_{c}, \psi^1_{c})$ is the knowledge game for $(A, \varphi, obs_c)$ and $(G^2_{c}, \psi^2_{c})$ is the knowledge game for $(G_{f}, \psi_f, obs_c)$. 
 %Then the state $v'$ is winning for Player I in $(G^2_{c}, \psi^2_{c})$ iff the state $\bigcup_{s' \in v'}s'$ is winning in for Player I in $(G^1_{c}, \psi)$.
 Then the relation $R = \{(v, v') | v = \bigcup_{s' \in v'}s'\}$ between the states of $G^1_{c}$ and $G^2_{c}$ is a bisimulation.
\end{theorem}
\begin{proof} 
 Suppose, that $G^1_{c} = (V, v_{init} , \Sigma, \rightarrow_g)$ and $G^2_{c}= (V', v'_{init} , \Sigma, \rightarrow'_g)$. 

 %We will assume for simplicity that $\rightarrow'_g$ and $\rightarrow_g$ are total in the sense that for any game states $v_{pred} \in V$, $v_{pred}' \in V'$ and a set of observable predicates $obs_{succ} \subseteq obs$ there are game states $v_{succ} \in V \cup \{\varnothing\}$ and $v'_{succ} \in V' \cup \{\varnothing\}$ such that $v_{pred} \xrightarrow{a}_g v_{succ}$, $v_{pred}' \xrightarrow{a}'_g v_{succ}'$ and $\gamma_{obs}(v_{succ}) = \gamma_{obs}(v_{succ}') = obs_{succ}$. The proof can be easily extended to the general case, when $\rightarrow_g'$ and $\rightarrow_g$ are not total.

% We show that the relation $R = \{(v, v') | v = \bigcup_{s' \in v'}s'\}$ is a bisimulation.
 First, we have $v_{init} = \{s_{init}\}$ and $v'_{init} = \{\{s_{init}\}\}$ and thus $(v_{init}, v'_{init}) \in R$.
 
Consider an action $a \in \Sigma$ and a pair of bisimilar states $v_{pred} \in V$ and $v_{pred}' \in V'$ (i.e. $(v_{pred}, v_{pred}') \in R$).
 It's easy to see, that $\gamma_{obs}(v'_{pred}) = \gamma_{obs}(v_{pred})$.
 Now we will demonstrate, that for any $a$-successor of $v_{pred}$ there is a bisimular $a$-successor of $v'_{pred}$ and vice versa.
 More precisely, we will show that for any game states $v_{succ} \in V$ and $v'_{succ} \in V'$, if  $\gamma_{obs}(v_{succ}) = \gamma_{obs}(v'_{succ})$ and there are transitions $v_{pred} \xrightarrow{a}_g v_{succ}$ and $v_{pred}' \xrightarrow{a}_g v'_{succ}$, then $v_{succ} = \bigcup_{v' \in v'_{succ}} v'$ and thus $(v_{succ}, v'_{succ}) \in R$. % (i.e. $v_{succ} = \bigcup_{v \in v_{succ}} v$).

Suppose, that $\gamma_{obs}(v_{succ})$ and $\gamma_{obs}(v'_{succ})$ are equal to $obs_{succ}$, and $\gamma_{obs}(v_{pred})$ and $\gamma_{obs}(v_{pred}')$ are equal to $obs_{pred}$.

First, it is easy to see that an LTS state $s$ belongs to the game state $v_{succ}$ iff there exists a sequence of transitions $s_1 \xrightarrow{a} s_2 \xrightarrow{a} \dots \xrightarrow{a} s_n \xrightarrow{a}_g s$, where $s_1 \in v_{pred}$, $\gamma_{obs}(s) = obs_{succ}$ and $\gamma_{obs}(s_i) = obs_{pred}$ for any $i \leq N$. 
We call such a sequence a proof sequence for $s$ in $G^1_{c}$.

Again, it's easy to see, that an LTS state $s'$ belongs to the set $\bigcup_{v' \in v'_{succ}} v'$ iff there exists a sequence of game transitions $v_1' \xrightarrow{a}'_g \dots \xrightarrow{a}'_g v_n' \xrightarrow{a}'_g v''$ such that $v'_1 \in v'_{pred}$, $s' \in v''$, $\gamma_{obs}(v'') = obs_{succ}$ and for any $i \leq N$ we have $\gamma_{obs}(v_i) = obs_{pred}$.
 The latter is true iff there exists a sequence of transitions 
\begin{eqnarray*}
   s_1^1 \xrightarrow{a} s_1^2 \xrightarrow{a} \dots \xrightarrow{a} s_1^{c_1} \xrightarrow{a} &  \\
   s_2^1 \xrightarrow{a} s_2^2 \xrightarrow{a} \dots \xrightarrow{a} s_2^{c_2} \xrightarrow{a} & \dots \\  
   s_n^1 \xrightarrow{a} s_n^2 \xrightarrow{a} \dots \xrightarrow{a} s_n^{c_n} \xrightarrow{a} & s'
\end{eqnarray*}
such that $s_1^1 \in \bigcup_{v' \in v_{pred}} v'$ and for any $i$, $j$ we have $\gamma_{obs}(s_i^j) = obs_{pred}$ and $\gamma_{obs'}(s_i^1) \neq \gamma_{obs'}(s_{i+1}^1)$.
We call such a sequence a proof sequence for $s'$ in $G^2_{c}$.
%$\gamma_{obs'}(s_i^j_1) \neq \gamma_{obs'}(s_{i+1}^j_2)$.

Now it's easy to see, that the sets $v_{succ}$ and $\bigcup_{v' \in v'_{succ}} v'$ coincide, since any proof sequence for $s$ in $G^2_{c}$ is a proof sequence for $s$ in $G^1_{c}$.
 At the same time, any proof sequence $s_1 \xrightarrow{a} s_2 \xrightarrow{a} \dots \xrightarrow{a} s_n \xrightarrow{a}_g s$ for $s$ in $G^1_{c}$ is a proof sequence for $s$ in $G^2_{c}$, if we compute the values of $c_i$ sequentially based on the changes of the $\gamma_{obs'}$ function. 
 More formally, we choose $c_k$ to be the largest integer such that $\gamma_{obs'}(s_{m+i}) = \gamma_{obs'}(s_{m+j})$ for $m = \sum_{l<k}c_l$ and any $i<c_k$, $j<c_k$. 

 This shows that $(v_{succ}, v'_{succ}) \in R$ and proves that $R$ is a bisimulation.
 \qed
\end{proof}

\setcounter{corollary}{0}
\begin{corollary} 
  Player I wins in $(G^1_{c}, \psi^1_{c})$ iff Player I wins in $(G^2_{c}, \psi^2_{c})$
\end{corollary}
\begin{proof}
According to the theorem~\ref{reusage_theorem} there is a bisimulation relation  $R = \{(v, v') | v = \bigcup_{s' \in v'}s'\}$ between $G^1_{c}$ and $G^2_{c}$.

It's easy to see that the bisimulation relation preserves the satisfiability of $\psi^1_{c}$ and $\psi^2_{c}$ formulas, i.e. for any pair of states $(v, v') \in R$, $v \models \psi^1_{c}$ iff $v' \models \psi^2_{c}$. 
Therefore Player I wins in $(G^2_{c}, \psi^2_c)$ iff Player I wins in $(G^1_{c}, \psi^1_c)$.

\qed

\end{proof}

\end{document}